%% file: main.tex
\documentclass[11pt]{article}
\usepackage[a4paper,
            left=1.25in,
            right=1.25in,
            top=1.25in,
            bottom=1.25in,
            footskip=.5in]{geometry}

\usepackage{authblk}
\usepackage{graphicx}
\usepackage{dsfont}
\usepackage{enumerate}
\usepackage{mathtools}

\usepackage[colorlinks=true,linkcolor=blue,citecolor=blue,urlcolor=blue,plainpages=false,pdfpagelabels]{hyperref}
\usepackage{amsmath}
\usepackage{amssymb}
\usepackage{bbm}
\usepackage{nicefrac}
\usepackage{url}
\usepackage{complexity}
\usepackage{xcolor}
\usepackage[normalem]{ulem}
\usepackage{amsthm}
\usepackage{appendix}

\usepackage{comment}

\usepackage[backend=biber, style=alphabetic, backref=true, hyperref=true, maxbibnames=99]{biblatex}
\usepackage[babel,english=british]{csquotes}
\DefineBibliographyStrings{english}{%
    backrefpage  = {cited on p.}, 
    backrefpages = {cited on pp.} 
}
\newcommand{\mcc}[1]{{\color[RGB]{0, 0, 205}{[MCC: #1]}}}

\newcommand{\dm}[1]{{\color[RGB]{150, 0, 150}{[DM: #1]}}}

\DeclareMathOperator{\Tr}{Tr}
\newcommand{\ket}[1]{|#1\rangle}
\newcommand{\bra}[1]{\langle #1|}

\newcommand{\stkout}[1]{\ifmmode\text{\sout{\ensuremath{#1}}}\else\sout{#1}\fi}
\newif\ifverbose
\verbosetrue

\newtheorem{theorem}{Theorem}[section]

\newtheorem{proposition}[theorem]{Proposition}
\newtheorem{assumption}[theorem]{Assumption}

\newtheorem{conjecture}[theorem]{Conjecture}
\newtheorem{lemma}[theorem]{Lemma}
\newtheorem{definition}[theorem]{Definition}

\theoremstyle{definition}

\newtheorem{remark}[theorem]{Remark}

\numberwithin{equation}{section}
\usepackage[capitalize,noabbrev,nameinlink]{cleveref}
\crefname{assumption}{Assumption}{assumptions}
\Crefname{Assumption}{Assumption}{Assumptions}

\input{commands.tex}

\renewcommand{\nicefrac}[2]{#1/#2}
\renewcommand{\tfrac}[2]{#1/#2}

\begin{document}

\title{Gibbs state postulate from dynamical stability }

\author[1,2,3]{Vjosa Blakaj\thanks{\href{mailto:vjosa.blakaj@tum.de}{vb@math.ku.dk}}}
\author[4]{Matthias C. Caro\thanks{\href{mailto:matthias.caro@warwick.ac.uk}{matthias.caro@warwick.ac.uk}}}
\author[5,6]{Anouar Kouraich\thanks{\href{mailto:anouar.kouraich@tum.de}{anouar.kouraich@tum.de}}}
\author[1]{Daniel Malz\thanks{\href{mailto:malz@math.ku.dk}{malz@math.ku.dk}}}
\author[5,6]{Michael M. Wolf\thanks{\href{mailto:m.wolf@tum.de}{m.wolf@tum.de}}}

\affil[1]{Department of Mathematical Sciences, University of Copenhagen, Copenhagen, Denmark}
\affil[2]{Max Planck Institute for the Science of Light, Erlangen, Germany}
\affil[3]{Department Physik, Friedrich-Alexander-Universität Erlangen-Nürnberg, Erlangen, Germany}
\affil[4]{Department of Computer Science, University of Warwick, Coventry, UK}
\affil[5]{Department of Mathematics, Technical University of Munich, Garching, Germany}
\affil[6]{Munich Center for Quantum Science and Technology (MCQST),  Munich, Germany}

\date{}
\setcounter{Maxaffil}{0}
\renewcommand\Affilfont{\itshape\small}

\maketitle

\begin{abstract}

 \noindent
Gibbs states play a central role in quantum statistical mechanics as the standard description of thermal equilibrium. Traditionally, their use is justified either by a heuristic, a posteriori reasoning, or by derivations based on notions of typicality or passivity. In this work, we show that Gibbs states are completely characterized by assuming dynamical stability of the system itself and of the system in weak contact with an arbitrary environment. This builds on and strengthens a result by Frigerio, Gorini, and Verri \cite{frigerio1986zeroth}, who derived Gibbs states from dynamical stability using an additional assumption that they referred to as the ``zeroth law of thermodynamics", as it concerns a nested dynamical stability of a triple of systems. We prove that this ``zeroth law" is redundant and that an environment consisting solely of harmonic oscillators is sufficient to single out Gibbs states as the only dynamically stable states.

    \noindent 
\end{abstract}

\section{Introduction}

In quantum statistical physics and quantum thermodynamics, the \emph{Gibbs state}
\begin{equation}
	\rho = \frac{\exp(-\beta  H)}{\Tr[\exp(-\beta H)]}\, ,
	\label{eq:gibbs}
\end{equation}
is used to describe a system in thermal equilibrium with a heat bath at inverse temperature $\beta$. Via the Gibbs state, macroscopic properties of the system can be computed from its microscopic description given by the Hamiltonian $H$. In this way, ~\cref{eq:gibbs} bridges the microscopic and macroscopic worlds while simultaneously connecting the description of static and dynamical properties. 

Various justifications and derivations of ~\cref{eq:gibbs} can be found in the literature. If we put aside arguments that trade between similar postulates (the Gibbs, maximum entropy, and minimum free energy postulates) one can identify three routes toward a derivation of Gibbs states: via typicality, passivity, or stability.

Using \emph{typicality} as the underlying fundamental concept is closest to the classical derivation of the Gibbs ensemble by Boltzmann~\cite{Boltzmann1877} and Gibbs~\cite{Gibbs1902} from the assumption that all microstates of a system plus reservoir at fixed energy are equally likely. Quantum versions of this derivation have been carried out successfully in \cite{CanonicalTypicality,Popescu2006}, where the starting point is a `typical' (in the sense of Haar-random) quantum state of a weakly coupled system plus reservoir at a given energy.

Derivations based on \emph{passivity} use a thermodynamic rather than a statistical viewpoint: the assumption that no work can be extracted from the state by cyclic operations. Along these lines, Pusz and Woronowicz~\cite{PuszWoronwicz} and Lenard~\cite{Lenard1978} derived the form of the Gibbs state and its generalization, the Kubo–Martin–Schwinger (KMS) state. In particular, \emph{complete passivity}---the assumption that no work can be extracted from arbitrarily many copies of a state---was shown to imply the Gibbs/KMS form.
Although remarkable, these conditions and derivations arguably do not provide an intuitive understanding of ``\emph{why nature produces Gibbs states}''---a question that calls for a more dynamical explanation. 

Inspired by earlier work by Haag, Kastler, and Trych-Pohlmeyer \cite{Haag1974}, Frigerio, Gorini, and Verri gave an elegant argument showing that the Gibbs state can be derived from a notion of \emph{dynamical stability}~\cite{gorini1985quantum,frigerio1986zeroth}. 
The underlying idea is that thermal states are expected to be dynamically stable, since systems are continuously weakly perturbed by their environment. Specifically, a system in equilibrium should not depart far from an equilibrium state due to a small perturbation of its Hamiltonian, a notion Frigerio, Gorini, and Verri call \emph{stability of order one}.
Second, they define a notion of mutual equilibrium: a system in a state $\rho$, governed by a Hamiltonian $H$, is in mutual equilibrium with system $(\rho',H')$ if and only if both remain approximately stationary when brought in contact. More precisely, a system $(\rho,H)$ is said to be \emph{stable of order two}, if, for any second system, governed by any $H'$, one can find a state $\rho'$ such that $(\rho\otimes\rho',H+H')$ is stable under small perturbations of the joint Hamiltonian.
The final assumption in~\cite{gorini1985quantum,frigerio1986zeroth} is \emph{stability of order three}, which the authors refer to as \emph{the zeroth law of thermodynamics} in this context. It requires that the state $\rho'$ can be chosen uniformly such that $(\rho\otimes\rho',H+H')$ is itself stable of order two, and this, they prove, implies that $\rho$ must be a Gibbs state. However, they write

\begin{center}
\parbox{0.8\textwidth}{
\textit{“one may wonder whether stability of order two alone is enough
to ensure that $\rho$ is the canonical ensemble: in other words, “does the zeroth 
law of thermodynamics hold by definition of mutual equilibrium?" [...]
We are inclined to believe that there are systems for which a state which is stable of order two need not be the canonical ensemble.” }\cite{frigerio1986zeroth}}
\end{center}

The main result of the present paper is that this is not the case: any state stable of order two must be a Gibbs state (\cref{theorem:second-order-stability-gibbs}).

This result has a number of interesting conceptual and technical implications.
On the conceptual side, it simplifies the set of assumptions necessary to derive thermodynamics from first principles.

An advantage compared to complete passivity is that one does not require arbitrarily many copies of the system.
Moreover, we will see that it is sufficient to demand stability when coupling to environments composed only of harmonic oscillators, something that naturally happens in many contexts. In fact, while we prefer individual harmonic oscillators over a field-theoretic framework for simplicity, the environment is, in essence, a single bosonic quantum field.

On the technical side, \cite{frigerio1986zeroth} show that stability of order two implies that the populations of the state form a monotonically non-increasing function of the energy, so in particular the state is passive. However, to go from mere passivity (a monotone population distribution) to a Gibbs state previously required either third-order stability \cite{frigerio1986zeroth} or an additional assumption \cite{PuszWoronwicz, Lenard1978}, for instance one based on the inclusion of arbitrarily many copies of the original system.

Our work complements the passivity-based characterizations of Gibbs states \cite{PuszWoronwicz, Lenard1978} by deriving the Gibbs form from dynamical stability, without the third-order stability assumption of Frigerio, Gorini and Verri \cite{frigerio1986zeroth}. A distinctive feature is that a universal class of auxiliary environments suffices: harmonic oscillators of variable frequencies---a single bosonic quantum field.

Before delving into the mathematical details, we note that redundancy of the zeroth law has also been demonstrated in phenomenological thermodynamics \cite{kammerlander2018zeroth}. Our result, however, is distinct, as it emerges from the framework of quantum statistical mechanics. And while the naming of the third hypothesis in \cite{frigerio1986zeroth} might be debatable, its redundancy renders this question of secondary importance.

\section{Preliminaries and Previous Results}

In this section, we collect some relevant mathematical concepts, introduce notation, and recall the results of Frigerio, Gorini, and Verri \cite{frigerio1986zeroth} on which our work is based.

\paragraph{Existence of Gibbs states}
The dynamics of the quantum systems under consideration, defined on a (typically infinite dimensional) Hilbert space $\mathcal{H}$, are governed by a Hamiltonian $H$ that, by assumption, has the property that all Gibbs states exist: 

\begin{assumption}[Hamiltonian]\label{Hamiltonianassumption}
All Hamiltonians used in this work are self-adjoint operators $H:\mathcal{D}(H) \subseteq \mathcal{H} \to \mathcal{H}$ assumed to be such that $\exp(-\beta H)$ is trace class for all inverse temperatures $\beta>0$. 
\end{assumption}

It is well known that this requires the Hamiltonian to be bounded from below and to have purely discrete spectrum \cite{LebLiebSimon}. The following proposition recalls this fact for completeness and pinpoints more precisely, which Hamiltonians satisfy Assumption \ref{Hamiltonianassumption} in infinite dimensions---those whose energy levels grow $E_n\nearrow\infty$ faster than logarithmic. 

\begin{proposition}[Existence of Gibbs states] Let $H$ be a self-adjoint operator on an infinite-dimensional Hilbert space. For every $\beta_0\in[0,\infty)$ the following are equivalent:
\begin{enumerate}
    \item[$(i)$]  The Gibbs operator $e^{-\beta H}$ is trace class for all $\beta>\beta_0$. That is,
    $$ Z_H(\beta)\coloneqq \Tr\!\bigl[\mathrm{e}^{-\beta H}\bigr] < \infty\qquad \forall
      \beta>\beta_0. $$
    \item[$(ii)$] $H$ is bounded from below, and its spectrum is purely discrete and consists of isolated eigenvalues of finite multiplicity that, when taking multiplicities into account and ordering $E_0 \leq E_1 \leq E_2 \leq ...$, satisfy $$   \limsup_{n\rightarrow\infty}\frac{\ln n}{E_n} \leq \beta_0.$$
\end{enumerate}
\end{proposition}
\begin{proof}
    $(i) \Rightarrow (ii):$ The assumption of $e^{-\beta H}$ being trace class implies that it is a compact operator and, in particular, bounded. Consider the spectrum of $H$ divided into an essential and a discrete part, $\sigma(H)=\sigma_{disc}(H)\cup\sigma_{ess}(H)$. By definition, the discrete part consists of isolated eigenvalues of finite multiplicity and without accumulation points. From continuous functional calculus (\cite{OperatorTheory}, Chap.5) we get the inclusion $\exp\big(-\beta\; \sigma_{ess}(H)\big)\subseteq \sigma_{ess}\big(\exp(-\beta H)\big)=\{0\}$ and thus $\sigma_{ess}(H)=\emptyset$ as $0$ is not in the range of the exponential function. Boundedness of the Gibbs operator by some constant $C$ on the level of eigenvalues then amounts to $\exp(-\beta E_n)\leq C$ so that $E_n\geq -\tfrac{1}{\beta}\ln C$, which shows that $H$ is bounded from below. 

    We aim for the contrapositive $\neg(ii) \Rightarrow \neg (i)$. So assume for some $p>1$ and $\beta>\beta_0$ that $\limsup_{n\rightarrow\infty} \tfrac{(\ln n)}{E_n}>p\beta$. By definition of the $\limsup$, there is a subsequence $(E_{n_k})_{k\in\mathbbm{N}}$ s.t. $\tfrac{(\ln n_k)}{E_{n_k}}>p\beta$, which means that $$n_k^{-\frac{1}{p}}<e^{-\beta E_{n_k}}.$$ 
    Using ordered  eigenvalues now allows us to bound
    $$\sum_{n=1}^{n_k} e^{-\beta E_n} \geq n_k e^{-\beta E_{n_k}} > n_k^{1-1/p}\longrightarrow\infty\quad\text{for}\quad k\rightarrow\infty.$$

    $(ii) \Rightarrow (i):$ Choose any $p\in(1,\beta/\beta_0)$. Then $\limsup_{n\rightarrow\infty} \tfrac{(\ln n)}{E_n}<\beta/p$, and for $n$ sufficiently large, we have $E_n>0$ and $\tfrac{(\ln n)}{E_n}<\beta/p $, which implies $\exp(-\beta E_n)<n^{-p}$. Since $p>1$, this yields a convergent series $\sum_n \exp(-\beta E_n) < \sum_n n^{-p} < \infty $.
\end{proof}

\paragraph{Stability under Perturbations and Thermal Contact}

\emph{Dynamical stability} is a fundamental characteristic of an equilibrium state. By this, we mean that an equilibrium state of a system with Hamiltonian $H$ should not change significantly in time when $H$ is perturbed by a small term $\lambda V$. If the state were to change appreciably, it could not describe an equilibrium situation, since under ordinary conditions it would be destroyed. A similar requirement applies when the system is brought into \emph{thermal contact} with another system: in equilibrium, the state should remain approximately stationary, regardless of the nature of the second system.

To formalize thermal contact and dynamical equilibrium, Frigerio, Gorini, and Verri \cite{frigerio1986zeroth} introduced a hierarchy of stability notions for a \emph{stationary state} $\rho$ of a system governed by a Hamiltonian $H$ fulfilling Assumption \ref{Hamiltonianassumption}. These increasingly stringent notions---\emph{first}, \emph{second}, and \emph{third-order stability}---are defined in \Cref{definition:stabilityOne,definition:stabilityTwo,definition:stabilityK}. While in \cite{frigerio1986zeroth} it is shown that third-order stability guarantees that $\rho$ is a Gibbs state, the question of whether second-order stability alone suffices remained open and will be answered affirmatively in \Cref{sec:mainresult}.\vspace{0.1cm}

We will denote by $\mathcal{S}(\mathcal{H})$ the set of density operators on the Hilbert space $\mathcal{H}$ and call a state $\rho\in\mathcal{S}(\mathcal{H})$ \emph{stationary}  under the dynamics generated by a Hamiltonian $H$ if $e^{-itH} \rho e^{itH} = \rho$ for all $t\geq 0$. All  perturbations will be assumed \emph{relatively bounded} with respect to the considered Hamiltonian:

\begin{definition}[Relatively bounded operators] 
    An operator $V$ on a Hilbert space $\mathcal{H}$ is called relatively bounded with respect to an operator $H$ (``$H$-bounded'') if the respective domains satisfy $\mathcal{D}(H) \subseteq \mathcal{D}(V)$ and if there exist constants $a, b \geq 0$ such that
    \begin{equation}
        ||V \psi || \leq a ||H \psi || + b ||\psi||\quad \textrm{for all } \psi \in \mathcal{D}(H)\, .
    \end{equation}
\end{definition}

With this we can introduce the first notion of stability:

\begin{definition}[Stability of order one~{\cite[Section 2]{frigerio1986zeroth}}]\label{definition:stabilityOne}
    Let $H:\mathcal{D}(H)\to\mathcal{H}$, with $\mathcal{D}(H)\subseteq\mathcal{H}$, be a Hamiltonian satisfying Assumption \ref{Hamiltonianassumption} and $\rho\in\mathcal{S}(\mathcal{H})$ a stationary state under the dynamics generated by $H$. We say that $\rho$ is \emph{stable of order one with respect to~$H$} if, for every $H$-bounded self-adjoint operator $V:\mathcal{D}(V)\to\mathcal{H}$, with $\mathcal{D}(V)\subseteq\mathcal{H}$ and for every bounded self-adjoint operator $O\in\mathcal{B}(\mathcal{H})$, 
    \begin{equation}
        \lim_{\lambda\to 0}\sup_{t\geq 0} \left(\Tr\left[O  e^{-i t (H+\lambda V)} \rho e^{i t (H+\lambda V)} \right] - \Tr\left[O \rho\right]\right)
        = 0\, .
    \end{equation}
\end{definition}\vspace{0.4cm}

The physical relevance of the above definition stems from the fact that the model Hamiltonian $H$ of a statistical mechanical system is rarely a complete and exact description of \emph{all} interactions influencing the time evolution of the system. As noted in \cite{frigerio1986zeroth}, the eigenstates of $H$ are ``not exact quantum states of a perfectly isolated system with every particle interaction fully accounted for''. Without stability, even ordinary circumstances could destabilize the system: weak residual interactions within the system itself may disrupt equilibrium, and interactions with the environment can likewise induce transitions between the eigenstates of $H$. The last considerations motivate the stronger notion of stability introduced below, which models thermal contact with an auxiliary system.

\begin{definition}[Stability of order two~{\cite[Section 3]{frigerio1986zeroth}}]\label{definition:stabilityTwo}
    Let $H:\mathcal{D}(H)\to\mathcal{H}$, with $\mathcal{D}(H)\subseteq\mathcal{H}$, be a Hamiltonian satisfying Assumption \ref{Hamiltonianassumption} and $\rho\in\mathcal{S}(\mathcal{H})$ a stationary state under the dynamics generated by $H$.   
    We say that $\rho$ is \emph{stable of order two with respect to~$H$} if, for every Hilbert space $\mathcal{H}'$ and for every Hamiltonian $H':\mathcal{D}(H')\to\mathcal{H}'$ satisfying Assumption \ref{Hamiltonianassumption}, there exists a state $\rho'\in\mathcal{S}(\mathcal{H}')$ such that $\rho\otimes\rho'\in \mathcal{S}(\mathcal{H}\otimes \mathcal{H}')$ is stable of order one with respect to~$H\otimes\mathds{1}_{\mathcal{H}'} + \mathds{1}_{\mathcal{H}}\otimes H'$.
\end{definition}\vspace{0.4cm}

In simpler words, the above definition tells us that in equilibrium the state $\rho$ remains stable under weak coupling to any environment that satisfies the same basic assumptions on the Hamiltonian $H'$. 
As noted in \cite{frigerio1986zeroth}, stability of order two implies stability of order one for the marginals (see \Cref{proposition:first-order-stability-product-to-factors1} for a formal statement and proof).

Along similar lines, stronger notions of stability have been introduced and were used  to characterize Gibbs states in \cite{frigerio1986zeroth}.

\begin{definition}[Stability of order three~{\cite[Section 4]{frigerio1986zeroth}}]\label{definition:stabilityK}
    Let $H:\mathcal{D}(H)\to\mathcal{H}$, with $\mathcal{D}(H)\subseteq\mathcal{H}$, be a Hamiltonian satisfying Assumption \ref{Hamiltonianassumption} and $\rho\in\mathcal{S}(\mathcal{H})$ a stationary state under the dynamics generated by $H$.
    We say that $\rho$ is \emph{stable of order three with respect to~$H$} if, for any Hilbert spaces $\mathcal{H}_1,\mathcal{H}_2$ and for any self-adjoint operators $H_i:\mathcal{D}(H_i)\to\mathcal{H}_i$ 
    satisfying Assumption \ref{Hamiltonianassumption}, there exist states $\rho_i\in\mathcal{S}(\mathcal{H}_i)$, $i=1,2$, such that $\rho\otimes\rho_1\otimes\rho_2$ is stable of order one with respect to~$H\otimes\mathds{1}_{\mathcal{H}_1}\otimes \mathds{1}_{\mathcal{H}_{2}} + \mathds{1}_{\mathcal{H}} \otimes H_1 \otimes \mathds{1}_{\mathcal{H}_2} + \mathds{1}_{\mathcal{H}} \otimes \mathds{1}_{\mathcal{H}_1} \otimes H_2 $, and the state $\rho_1 \in \mathcal{S}(\mathcal{H}_1)$ is independent of the choice of $\mathcal{H}_2$, $H_2$, and $\rho_2$.
\end{definition}

Stability of order $k$ for any $k\in\mathbb{N}$ is defined analogously.
An equivalent recursive definition of stability of order $k \in \mathbb{N}$ is as follows: $\rho$ is stable of order $k$ with respect to~$H$ if, for any Hilbert space $\mathcal{H'}$ and for any Hamiltonian $H':\mathcal{D}(\mathcal{H}')\to\mathcal{H}'$ satisfying Assumption \ref{Hamiltonianassumption}, there exists a state $\rho'\in\mathcal{S}(\mathcal{H}')$ such that $\rho\otimes\rho'$ is stable of order $k-1$ with respect to~$H\otimes\mathds{1}_{\mathcal{H}'} + \mathds{1}_{\mathcal{H}}\otimes H'$.

Note that by the nested logic of the form
$$\forall H_1\exists\rho_1\forall H_2\exists\rho_2\ldots \qquad \rho\otimes\rho_1\otimes\rho_2\ldots \text{is stable of order one},$$
stability of order $k$ is more demanding than stability of order $k-1$, as it requires that each $\rho_i$ must not depend on any $H_j, \rho_j$ with $j>i$.
\vspace{0.2cm}

Dynamical stability of a state with respect to the Hamiltonian imposes structural constraints on the state. As shown in \cite{frigerio1986zeroth}, stability of order one already implies that the state is a function of the Hamiltonian, while stability of order two further requires this function to be non-increasing, so that the state's population does not increase as the energy increases. These results are summarized in the following theorem.

\begin{theorem}[Structural constraints from dynamical stability {\cite[Theorems 1 and 2]{frigerio1986zeroth}}]\label{theorem:first-seecond-order-stability-function}
    Let $H:\mathcal{D}(H)\to\mathcal{H}$, with $\mathcal{D}(H)\subseteq\mathcal{H}$, be a Hamiltonian satisfying Assumption \ref{Hamiltonianassumption}, and let $\rho\in\mathcal{S}(\mathcal{H})$ be a quantum state. Then:
    \begin{enumerate}
        \item If $\rho$ is stable of order one with respect to~$H$, there exists a function $f:\sigma(H)\to [0,1]$ such that $\rho=f(H)$.
        \item If $\rho$ is stable of order two with respect to~$H$, the function $f$ is monotonically non-increasing on $\sigma(H)$, and $f\!\restriction_{f^{-1}((0,1]]}$ is strictly decreasing.
    \end{enumerate}
\end{theorem}

The following remark illustrates how second-order stability ensures that equal energy gaps in the Hamiltonians $H$ and $H'$ lead to identical population ratios.

\begin{remark}[{\cite[Section 3]{frigerio1986zeroth}}]\label{Remark:GapInterpretation} 
Let $\rho \in \mathcal{S}(\mathcal{H})$ be stable of order two with respect to the Hamiltonian $H$. From \Cref{theorem:first-seecond-order-stability-function} and \Cref{proposition:first-order-stability-product-to-factors1} it follows that, for every Hamiltonian $H'$, there exist a state $\rho'$, a function $h : \sigma(H) + \sigma(H') \to [0,1]$, a non-increasing function $f : \sigma(H) \to [0,1]$, and a function $g : \sigma(H') \to [0,1]$ such that
\begin{equation}
    \rho \otimes \rho' = h\bigl(H \otimes \mathds{1}' + \mathds{1} \otimes H'\bigr)
    = f(H) \otimes g(H').
\end{equation}
In particular, this implies that
\begin{equation}
    \frac{p_n}{p_m} = \frac{p'_s}{p'_r}
    \quad \text{whenever} \quad
    E_n - E_m = E'_s - E'_r.
\end{equation}
By \cref{proposition:first-order-stability-product-to-factors1}, we know that $f(E_n) =: p_n \in \sigma(\rho)$ and $f(E_m) =: p_m \in \sigma(\rho)$, as well as $g(E'_s) =: p'_s \in \sigma(\rho')$ and $g(E'_r) =: p'_r \in \sigma(\rho')$.
This relation highlights a correspondence between equal energy gaps in the Hamiltonians $H$ and $H'$ and the equality of the corresponding probability ratios. 
\end{remark}

Another important observation concerns the form of the state $\rho' \in \mathcal{S}(\mathcal{H})$ under stability of order two when the auxiliary Hamiltonian $H'$ is chosen to be that of a harmonic oscillator.

\begin{remark}\label{remark:Stability2withHOgivesGibbs}
    Let $\rho \in \mathcal{S}(\mathcal{H})$ be stable of order two with respect to the Hamiltonian $H$, and choose the auxiliary Hamiltonian $H'$ to be that of a harmonic oscillator. Then the corresponding state $\rho'$ is necessarily a thermal (Gibbs) state. This follows from \cref{Remark:GapInterpretation} {\cite[Theorem 2]{frigerio1986zeroth}}; see also the proof of \cref{theorem:second-order-stability-gibbs}.    
\end{remark}

First-order stability is not sufficient to single out the Gibbs state among all functions of the Hamiltonian. Indeed, by the sufficient criterion for first-order stability in \cite[Theorem~1]{frigerio1986zeroth}, any non-negative function that decays at least exponentially fast already defines a first-order stable state. By contrast, it is shown in \cite{frigerio1986zeroth} that stability of order three enforces the Gibbs form:

\begin{theorem}[{\cite[Theorem 3]{frigerio1986zeroth}}]\label{theorem:third-order-stability-gibbs}
    Let $H:\mathcal{D}(H)\to\mathcal{H}$, with $\mathcal{D}(H)\subseteq\mathcal{H}$, be a Hamiltonian satisfying Assumption \ref{Hamiltonianassumption} and $\rho\in\mathcal{S}(\mathcal{H})$ a stable state of order three with respect to~$H$.
    Then, there exists $\beta\in (0,\infty]$ such that
    \begin{equation}
        \rho
        = \rho(\beta) 
        = \frac{e^{-\beta H}}{\Tr[e^{-\beta H}]}\, .
    \end{equation}
    Here, taking $\beta = \infty$ (zero temperature) gives the special case where $\rho$ is proportional to the ground-state space projector.
\end{theorem}

Since Gibbs states are stable of all orders $k \in \mathbb{N}$, stability of order three already characterizes them uniquely \cite{frigerio1986zeroth}.  In the next section, we demonstrate that placing the system in thermal contact with another system is sufficient to fix the temperature and enforce the Gibbs form, without invoking the stability of order three.

\section{Stability of Order Two Characterizes Gibbs States}\label{sec:mainresult}

This section contains the central result of our work: a stationary quantum state that remains stable under weak coupling to any auxiliary system in equilibrium must be of Gibbs form. In fact, as can be seen in the proof, it already suffices to require this stability only for auxiliary systems that are \emph{harmonic oscillators}.

\begin{theorem}\label{theorem:second-order-stability-gibbs}
    Let $H:\mathcal{D}(H)\to\mathcal{H}$, with $\mathcal{D}(H)\subseteq\mathcal{H}$, be a Hamiltonian satisfying Assumption \ref{Hamiltonianassumption} and $\rho\in\mathcal{S}(\mathcal{H})$ stable of order two with respect to~$H$. Then, there exists an inverse temperature $\beta \in (0, \infty]$ such that
    \begin{equation}
        \rho
        = \rho(\beta)
        = \frac{e^{-\beta H}}{\Tr[e^{-\beta H}]} \, .
    \end{equation}
    Here, taking $\beta = \infty$ (zero temperature) gives the special case where $\rho$ is proportional to the ground-state space projector.
\end{theorem}
\begin{proof} We denote by $E_0<E_1<E_2<\ldots$ the distinct eigenvalues of $H$.
    As $\rho \in\mathcal{S}(\mathcal{H})$ is stable of order two with respect to~$H$, by \Cref{theorem:first-seecond-order-stability-function} there is a non-increasing function $f:\sigma(H)\to [0,1]$ such that $\rho = f(H)$.
    Therefore, we can define, for every $n\in\mathbb{N}$,
    \begin{equation}
        \beta_n
        = -\frac{\log(\nicefrac{f(E_n)}{f(E_0)})}{E_n-E_0}\in (0,\infty],
        \label{eq:beta}
    \end{equation}
    such that by definition
    \begin{equation}
        f(E_n)
        = f(E_0)\exp\left(-\beta_n (E_n-E_0)\right)\quad\forall n\in\mathbb{N}_0. 
    \end{equation}
    Thus, it suffices to prove that $\beta_n=\beta_m$ holds for all $n,m\in\mathbb{N}$.

    \vspace{0.2cm}

    So, let us consider two specific $ n,m\in\mathbb{N}$ with $n\neq m$. Then in particular $E_m,E_n > E_0$ and we can assume w.l.o.g. that $\beta_n\geq\beta_m$. We want to show that $\beta_n=\beta_m$.\vspace{0.3cm}

In the definition of stability of order two (\Cref{definition:stabilityTwo}), the auxiliary Hamiltonian $H'$ may be chosen as any self-adjoint operator satisfying Assumption \ref{Hamiltonianassumption}. We take $H'$ to be a three-mode harmonic oscillator,
\begin{equation}
    H' \coloneqq A_{\omega_1} + A_{\omega_2} + A_{\omega_3}\, ,
\end{equation}
where, for $i=1,2,3$, $A_{\omega_i}$ denotes the Hamiltonian of a single-mode harmonic oscillator with frequency $\omega_i$. We choose the frequencies
\begin{equation}
    \omega_1 \coloneqq E_n - E_0, \qquad
    \omega_2 \coloneqq E_m - E_0, \qquad
    \omega_3 \coloneqq -p\,(E_n - E_0) + q\,(E_m - E_0),
\end{equation}
where $p, q \in \mathbb{N}$ will be specified later in the proof in such a way that $\omega_3>0$. Let $\{\ket{n_1,n_2,n_3} : n_1,n_2,n_3 \in \mathbb{N}_0\}$ denote the corresponding number basis. 
On this basis, each mode acts as
\begin{equation}
    A_{\omega_i} \ket{n_1,n_2,n_3}
    = \omega_i\left(n_i + \tfrac{1}{2}\right)\ket{n_1,n_2,n_3},
    \qquad i = 1,2,3,
\end{equation}
and the spectrum of $H'$ is the set 
\begin{equation}
    \sigma(H') = \{E'_{i,j,k} \coloneqq \omega_1 (i+1/2) + \omega_2(j+1/2)+\omega_3(k+1/2)\, |\, i,j,k \in \mathbb{N}_0\}.
\end{equation}
As $\rho$ is stable of order two, there exists a state $\rho_\omega$ such that $\rho \otimes \rho_\omega$ is stable of order one with respect to $H\otimes \mathds{1}_{\mathcal{H}'} + \mathds{1}_{\mathcal{H}}\otimes H'$. We argue that $\rho_\omega$ must then be a tensor product of Gibbs states:
Indeed, with $f,g,h$ as in \cref{theorem:first-seecond-order-stability-function} and \cref{Remark:GapInterpretation}, we have
\begin{equation}
    \rho \otimes \rho_\omega
    = f(H) \otimes g(H')
    = h\bigl(H + A_{\omega_1} + A_{\omega_2} + A_{\omega_3}\bigr)\, .
\end{equation}

Since $E_n - E_0 = \omega_1$, the eigenvalues $E'_{j,k,l}, E'_{j+1,k,l} \in \sigma(H')$ satisfy
\begin{equation}
    E_n + E'_{j,k,\ell} = E_0 + E'_{j+1,k,\ell}\, ,
\end{equation}
and hence
\begin{align}\label{eq:fg}
    &f(E_n)\,
      g\Bigl(\omega_1\!\left(j + \tfrac{1}{2}\right)
            + \omega_2\!\left(k + \tfrac{1}{2}\right)
            + \omega_3\!\left(\ell + \tfrac{1}{2}\right)\Bigr)\notag\\
    &\qquad= f(E_0)\,
      g\Bigl(\omega_1\!\left(j + \tfrac{3}{2}\right)
            + \omega_2\!\left(k + \tfrac{1}{2}\right)
            + \omega_3\!\left(\ell + \tfrac{1}{2}\right)\Bigr)\, .
\end{align}
We now define
\begin{equation}
    p_{j,k,\ell} \coloneqq 
    g\bigl(\omega_1(j+\tfrac{1}{2}) + \omega_2(k+\tfrac{1}{2}) + \omega_3(\ell+\tfrac{1}{2})\bigr)\, .
\end{equation}
From \cref{eq:fg}, for any $j \geq 0$ we obtain
\begin{align}
    p_{j+1,k,\ell} 
    &= p_{j,k,\ell}\,\frac{f(E_n)}{f(E_0)}
     = p_{j,k,\ell}\,e^{-\beta_n(E_n - E_0)} \notag\\
    &= \ldots = p_{0,k,\ell}\,e^{-(j+1)\beta_n(E_n - E_0)}.
\end{align}
Similarly, using that $E_m-E_0=\omega_2$, we can show
\begin{align}
    p_{j,k+1,\ell} 
    &= p_{j,k,\ell}\,\frac{f(E_m)}{f(E_0)}
     = p_{j,k,\ell}\,e^{-\beta_m(E_m - E_0)} \notag\\
    &=\ldots = p_{j,0,\ell}\,e^{-(k+1)\beta_m(E_m - E_0)}.
\end{align}
For our third recursion relation, since $\omega_3 = q\omega_2 - p\omega_1 (>0)$, we get that 
\begin{equation}
    p_{j+p,k,\ell+1} = p_{j,k+q,\ell}\, .
\end{equation}

Using the three recursive relations, we can now rewrite the eigenvalues as a recursion in the third index:
\begin{align}
    p_{j,k,\ell+1}
    &= p_{j+p,k,\ell+1}\,e^{\beta_n p(E_n - E_0)} \nonumber\\
    &= p_{j,k+q,\ell}\,e^{\beta_n p(E_n - E_0)} \nonumber\\
    &= p_{j,k,\ell}\,e^{\beta_n p(E_n - E_0) - \beta_m q(E_m - E_0)}\label{eq:pjkl}
\end{align}
In particular,
\begin{equation}
    p_{j,k,\ell}
    = p_{j,k,0}\,e^{\ell\bigl(\beta_n p\omega_1 - \beta_m q\omega_2\bigr)}.
\end{equation}
Combining all of the above, we obtain
\begin{equation}
    p_{j,k,\ell}
    = p_{0,0,0}\,
      e^{-\beta_n j\omega_1}\,
      e^{-\beta_m k\omega_2}\,
      e^{-\ell\bigl(-\beta_n p\omega_1 + \beta_m q\omega_2\bigr)}.
\end{equation}
Hence $\rho_\omega$ is the tensor product of two Gibbs states with inverse temperatures
$\beta_n$ and $\beta_m$, and a third Gibbs state with an effective inverse temperature that depends linearly on $\beta_n$ and $\beta_m$.\vspace{0.1cm}

To complete the proof for $\beta_n=\beta_m$, we show that the assumption $\beta_n>\beta_m$ contradicts the trace class property of $\rho_\omega$. To this end, we choose integers $p,q > 0$ such that
    \begin{equation}
        1 < \frac{\omega_2}{\omega_1}\frac{q}{p} < \frac{\beta_n}{\beta_m}.
    \end{equation}
Then
\begin{equation}
     \beta_n p \omega_1 > \beta_m q \omega_2
    \quad\text{and}\quad
    0 < q\omega_2 - p\omega_1.
\end{equation}
Equivalently,
\begin{equation}
     s \coloneqq e^{-\beta_m q\omega_2 + \beta_n p\omega_1} > 1.
\end{equation}
From the recursion in the third index we obtain
\begin{equation}
     p_{0,0,\ell} = p_{0,0,0}\,s^{\ell}, \qquad \ell \ge 0.
\end{equation}
Since $s>1$, the sequence $(p_{0,0,\ell})_{\ell \ge 0}$ is not summable, which contradicts the trace-class (summability) condition for $\rho_\omega$.
Therefore, we conclude that $\beta_n = \beta_m \equiv \beta$ for all $n,m\in \mathbb{N}$, and hence $\rho = f(H) \propto e^{-\beta H}$.
\end{proof}

\begin{remark}[Harmonic oscillators of various frequencies suffice]
    Although \cref{theorem:second-order-stability-gibbs} and the \cref{definition:stabilityOne,definition:stabilityTwo} are formulated  with respect to arbitrary auxiliary systems, as a consequence of the proofs of \Cref{theorem:first-seecond-order-stability-function,theorem:second-order-stability-gibbs} it is indeed sufficient to restrict to auxiliary harmonic oscillators of various frequencies. So effectively, the considered environment is a bosonic quantum field.
\end{remark}

\begin{remark}[One- and two-mode harmonic oscillators do not suffice]
    The proof of \cref{theorem:second-order-stability-gibbs} relies on dynamical stability when the system is in weak contact with any environment system that is described by  three-mode quantum harmonic oscillators. This means that the considered perturbations simultaneously affect three frequencies. One may wonder whether auxiliary single-mode quantum harmonic oscillators already suffice here, as they do for \cite{frigerio1986zeroth}'s proofs of \Cref{theorem:first-seecond-order-stability-function,theorem:third-order-stability-gibbs}. 
    In \Cref{appendix:mutually-rational}, we show that this is sufficient under the additional assumption that all energy gaps of the system Hamiltonian are commensurable.
    However, demanding stability with respect to environments $H'$ containing fewer than three-mode harmonic oscillators does not always force $\rho$ to be the Gibbs state, as we show with a counterexample below. To show that our counterexample is stable, we use the partial converse in \cite[Theorem 1]{frigerio1986zeroth}, which states: If $\rho$ is an exponentially decaying function of $H$, then $\rho$ is stable of order $1$ with respect to~$H$.

Suppose the system in question is three-dimensional and has incommensurable energy gaps, say, $(E_2-E_0)/(E_1-E_0)\notin\mathbb Q$, and we attach to it a single harmonic oscillator $A_\omega$.
The spectrum of $H+A_\omega$ is $\{E_j+\omega(k+1/2)~|~j\in\{0,1,2\},k\in\mathbb N_0\}$.
Note that it is impossible to choose $\omega$ to be an integer fraction of both $E_2-E_0$ and $E_1-E_0$.
If it is not an integer fraction of either, the spectrum has no degeneracies, and $\rho\otimes\rho'$ is stable with $\rho'=g(A_\omega)$ as long as $g$ decays exponentially with energy.
If it is an integer fraction of one transition, say $E_1-E_0$, there are infinitely many degeneracies of the form $E_1+k\omega=E_0+(k-\ell)\omega$ for some fixed $\ell=(E_1-E_0)/\omega$, but none that involve $E_2$.
If we now take $p'_k\propto\exp(-\beta_1k\omega)$, we find that again $\rho\otimes\rho'$ is an asymptotically exponentially decaying function of energy and thus stable.
The same works when exchanging $1\leftrightarrow2$.
Thus, in fact, for any $\rho=f(H)$ and any $\omega$, we can always find $\rho'$ such that the joint system is stable.

The argument in the previous paragraph can be extended to environments composed of exactly two harmonic modes $A_{\omega_1}+A_{\omega_2}$.
First note that in the two-oscillator setting, we can restrict ourselves to considering incommensurable $\omega_1, \omega_2$, since otherwise we could instead equivalently consider a single harmonic oscillator with a frequency equal to the greatest common measure of $\omega_1,\omega_2$, which reduces to the case considered in the previous paragraph.

We will argue that we can always choose
\begin{equation}
    \rho'=g(A_{\omega_1}+A_{\omega_2})
\end{equation}
such that the overall state $\rho\otimes\rho'$ is stable.
 
    Let $\vec\Omega=(\omega_1,\omega_2)$. Suppose there exist $\vec x,\vec y\in\mathbb Z^2$ such that $\vec x\cdot\vec\Omega=E_1-E_0$ and $\vec y \cdot \vec \Omega=E_2-E_0$.
    Note that $\vec x$ and $\vec y$ are linearly independent since the gaps are incommensurable by assumption.
    Next, we pick the smallest set of vectors $\mathcal A\subset\mathbb N^2$
    such that the following set 
    \begin{equation}
        \mathcal B = \{\vec a+z_1 \vec x+ z_2 \vec y\,|\,z_1,z_2\in\mathbb Z,\vec a\in\mathcal A\}
    \end{equation}
    contains $\mathbb N^2$ as a proper subset, $\mathbb N^2\subset\mathcal B$.
    Due to the linear independence of $\vec x,\vec y$, we can thus write any $\vec v\in\mathbb N^2$ uniquely as $\vec v=\vec a+z_1\vec x+z_2\vec y$.
    Let us now specify function values $g(\vec a\cdot\vec\Omega)$ for all $\vec a\in\mathcal A$. The choice is almost arbitrary, but we take them as some exponentially decaying function of their corresponding energy.
    Stability then forces us to set all the other populations according to
    \begin{equation}
        g((\vec a+z_1\vec x+z_2\vec y)\cdot\vec\Omega)=g(\vec a\cdot\vec\Omega)(p_1/p_0)^{z_1}(p_2/p_0)^{z_2}.
    \end{equation}
    As a last step, we can normalize $g$ to unit trace.
    Then, by construction, $\rho\otimes\rho'$ is stable since it is an asymptotically exponentially decaying function of energy.

\end{remark}

\section*{Acknowledgments} 
VB acknowledges support from the European Research Council (ERC, Grant Agreement No.~818761), VILLUM FONDEN via the QMATH Centre of Excellence (Grant No.~10059), and the Novo Nordisk Foundation (Grant No.~NNF20OC0059939, “Quantum for Life”). MCC was partially supported by a DAAD PRIME fellowship. DM acknowledges financial support by the Novo Nordisk Foundation under grant numbers NNF22OC0071934 and NNF20OC0059939. MMW and AK acknowledge support by the Deutsche Forschungsgemeinschaft (DFG, German Research Foundation) under Germany's Excellence Strategy –  EXC-2111 – 390814868 and via the TRR 352 – Project-ID 470903074. \vspace{0.1cm}

VB and MCC thank Jonathan Oppenheim, and VB and AK thank Andreas Winter for helpful discussions. VB also thanks Álvaro Alhambra, Daniel Burgarth, Toby Cubitt, Marius Junge and Cambyse Rouzé.\vspace{1cm}

\textbf{Data Availability} Data sharing is not applicable to this article as no datasets were generated or analysed
during the current study.\vspace{0.5cm}

\textbf{Conflict of Interest} The authors have no conflict of interest to declare that are relevant to the content of this
article.

\newpage
\setcounter{secnumdepth}{0}
\defbibheading{head}{\section{References}}
\sloppy
\printbibliography[heading=head]

\newpage
\appendix
\setcounter{secnumdepth}{2}

\section{The Case of Commensurable Energy Gaps}\label{appendix:mutually-rational}

In this Appendix, we show that if all energy gaps of $H$ are commensurable, then dynamical stability under weak contact with single-mode quantum harmonic oscillators already implies the Gibbs form. Although one could in this case follow the same line of reasoning as in the proof of \cref{theorem:second-order-stability-gibbs}, we instead present a different, elegant proof, relying on the following lemma, which we obtain by extending the argument of \cite[Section 3]{frigerio1986zeroth}.


\begin{lemma}[{\cite{thirring1983course}, \cite[Section 3]{frigerio1986zeroth}}]\label{lemma:defining-phi}
    Let $H:\mathcal{D}(H)\to\mathcal{H}$, with $\mathcal{D}(H)\subseteq\mathcal{H}$, be a Hamiltonian satisfying Assumption \ref{Hamiltonianassumption}, and let $\rho\in\mathcal{S}(\mathcal{H})$ be stable of order two with respect to~$H$.
    Let $\mathcal{H}'$ be another Hilbert space, and let $H':\mathcal{D}(H')\to\mathcal{H}'$, with $\mathcal{D}(H')\subseteq\mathcal{H}'$, be a Hamiltonian also satisfying Assumption \ref{Hamiltonianassumption}. Suppose $\rho'\in\mathcal{S}(\mathcal{H}')$ is such that the product state $\rho\otimes\rho'$ is stable of order one with respect to the total Hamiltonian~$H\otimes\mathds{1}_{\mathcal{H}'} + \mathds{1}_{\mathcal{H}}\otimes H'$.
    Assume further that $\rho = f(H)$ and $\rho'=g(H')$, where the functions $f$ and $g$ are given according to \Cref{Remark:GapInterpretation}.
    Then, for any fixed $r\in\mathbb{N}_0$ such that $g(E'_r)>0$, we can define a function $\Phi:(\sigma(H) - E_0)+(\sigma(H') - E'_r)\to [0,\infty)$ that relates the populations of different energy levels via
    \begin{equation}
        \Phi(E_m - E_0+ E'_n - E'_r)
        = \frac{f(E_m) g(E'_n)}{f(E_0) g(E'_r)}\, .
    \end{equation}
    This function $\Phi$ satisfies
    \begin{equation}\label{eq:functional-equation}
        \Phi(E - E_0 + E' - E'_r)
        = \Phi(E - E_0)  \Phi(E' - E'_r)\quad\forall E\in\sigma(H), \forall E'\in\sigma(H') \,
    \end{equation}
    and moreover $\Phi\rvert_{(\sigma(H) - E_0)}$ is independent of $H'$.
\end{lemma}
\begin{proof}
    Since $\rho'$ is a quantum state, there exists $r\in\mathbb{N}_0$ with $g(E'_r)>0$.
    Now recall from \Cref{Remark:GapInterpretation} that
    \begin{equation}
        \frac{p_n}{p_m} = \frac{p'_s}{p'_r} \quad \text{whenever} \quad
        E_n - E_m = E'_s - E'_r \, .
    \end{equation}
    Thus, $\Phi$ is well-defined, and the Cauchy functional equation \cref{eq:functional-equation} is immediate from the definition of $\Phi$.
\end{proof}

We are now ready to give a simple proof of \Cref{theorem:second-order-stability-gibbs} restricted to Hamiltonians with commensurable energy gaps. Although this proof does not establish the full result, it is conceptually different from the full proof of \Cref{theorem:second-order-stability-gibbs}; in particular, it uses only auxiliary single-mode quantum harmonic oscillators.\vspace{0.1cm}

The argument is as follows.

\hspace{0.1cm} First, assume without loss of generality that $E_0=0$. Indeed, we can always shift the Hamiltonian by a constant and consider $H - E_0 \mathds{1}$ instead of $H$. Such a shift only affects the unitary evolution $e^{-itH}$ by a global phase, leaving the dynamics---and hence the definition of stability---unchanged.

Assume now that $\rho$ is stable of order two with respect to~$H$. 
Define $\beta_n$ for $n\in\mathbb{N}$ as in the proof of \Cref{theorem:second-order-stability-gibbs}
to enforce
\begin{equation}
    f(E_n)
    = f(E_0)\exp\left(-\beta_n E_n\right)\quad\forall n\in\mathbb{N}_0\, ,
\end{equation}
independently of which value we choose for $\beta_0$.
Thus, it suffices to prove that $\beta_m=\beta_n$ holds for all $m,n\in\mathbb{N}$.

Let $m,n\in\mathbb{N}$ with $m<n$. In particular $E_m,E_n > 0$. 
Now assume that $E_m(=E_m-E_0)$ and $E_n(=E_n-E_0)$ are commensurable, that is, $\tfrac{E_m}{E_n}\in\mathbb{Q}$.

Take coprime $k,\ell \in \mathbb{N}$ such that $\tfrac{E_m}{E_n} = \tfrac{k}{\ell}$ and choose $H'$ to be the Hamiltonian of a harmonic oscillator with frequency $\omega := \frac{E_n}{\ell} = \frac{E_m}{k}\,$. Then, 
\begin{equation}
    \Phi(k E_n)
    = \Phi(\underbrace{E_n}_{\in\sigma(H)} + \underbrace{(k-1)E_n}_{=\frac{(k-1)\ell E_n}{\ell} \in\sigma(H')-E'_r})
    \overset{\Cref{eq:functional-equation}}{=} \Phi(E_n)\Phi((k-1)E_n)\, .
\end{equation}
By iterating we obtain $\Phi(k E_n) = \Phi(E_n)^k$. Similarly $\Phi(\ell E_m)=\Phi(E_m)^\ell$.
As $k E_n = \ell E_m$, it follows that
\begin{equation}
    \Phi(E_m)
    = \Phi(E_n)^{\nicefrac{k}{\ell}}\, .
\end{equation}
Using $\Phi(E_j) = \exp(-\beta_j E_j)$, we get
\begin{equation}
     \exp(-\beta_m E_m)
    = \left(\exp(-\beta_n E_n)\right)^{k/\ell}
    = \exp(-\beta_n E_m)\, ,
\end{equation}
and hence $\beta_m = \beta_n$ because $E_m > 0$.

\section{Stability of Order Two Implies Stability of Order One for the Marginals}

The observation made in this part of the Appendix is already noted in \cite{frigerio1986zeroth}; we nevertheless include a simple proof to keep the presentation self-contained and coherent.

\begin{proposition}[Stability of order two implies stability of order one for the marginals]\label{proposition:first-order-stability-product-to-factors1}
    Let $H:\mathcal{D}(H)\to\mathcal{H}$, with $\mathcal{D}(H)\subseteq\mathcal{H}$, and $H':\mathcal{D}(H')\to\mathcal{H}'$, with $\mathcal{D}(H')\subseteq\mathcal{H}'$, be Hamiltonians satisfying Assumption \ref{Hamiltonianassumption}. Assume that $\rho\in\mathcal{S}(\mathcal{H})$ and $\rho'\in\mathcal{S}(\mathcal{H}')$ are such that $\rho\otimes\rho'$ is stable of order one with respect to $H\otimes\mathds{1}_{\mathcal{H}'} + \mathds{1}_{\mathcal{H}}\otimes H'$.
    Then $\rho$ is stable of order one with respect to $H$ and $\rho'$ is stable of order one with respect to $H'$.
\end{proposition}

\begin{proof}
We show the claim for $(\rho,H)$; the case $(\rho',H')$ is analogous.\vspace{0.1cm}

By assumption, $\rho \otimes \rho'$ is stable of order one with respect to 
$H\otimes\mathds{1}_{\mathcal{H}'} + \mathds{1}_{\mathcal{H}}\otimes H'$, that is, for any 
$(H\otimes\mathds{1}_{\mathcal{H}'} + \mathds{1}_{\mathcal{H}}\otimes H')$-bounded self-adjoint operator $V$ and any self-adjoint operator
$O \in \mathcal{B}(\mathcal{H} \otimes \mathcal{H}')$,
\begin{equation}\label{limitorder21}
        \lim_{\lambda\to 0}\sup_{t\geq 0} \Tr\!\left[O \Bigl( e^{-i t (H\otimes\mathds{1}_{\mathcal{H}'} + \mathds{1}_{\mathcal{H}}\otimes H' + \lambda V)} \,\rho \otimes \rho' \,e^{i t (H\otimes\mathds{1}_{\mathcal{H}'} + \mathds{1}_{\mathcal{H}}\otimes H'+\lambda V)} - \rho \otimes \rho' \Bigr)\right] 
        = 0\, .
\end{equation}    

In particular, we may choose $O=\tilde O\otimes\mathds{1}$ and $V=\tilde V\otimes\mathds{1}$, where $\tilde O\in\mathcal{B}(\mathcal{H})$ is any self-adjoint operator and $\tilde V$ is any self-adjoint $H$-bounded operator on $\mathcal{H}$. Under these choies \eqref{limitorder21} reduces to
\begin{equation}
    \lim_{\lambda\to0}\sup_{t\ge0}
\Tr\!\left[\tilde O\,\left(e^{-it(H+\lambda\tilde V)}\rho\,e^{it(H+\lambda\tilde V)}
-\rho \right) \right]
=0.
\end{equation}

Since $\tilde O$ and $\tilde V$ were arbitrary, this is exactly stability of order one of $\rho$ with respect to $H$.
\end{proof}

\end{document}

%% file: commands.tex
\renewcommand{\epsilon}{\varepsilon}
\renewcommand{\epsilon}{\varepsilon}

\newcommand{\e}{\mathrm{e}}



%% file: gibbs.bib
@article{frigerio1986zeroth,
  title={The zeroth law of thermodynamics},
  author={Frigerio, Alberto and Gorini, Vittorio and Verri, Maurizio},
  journal={Physica A: Statistical Mechanics and its Applications},
  volume={137},
  number={3},
  pages={573--602},
  year={1986},
  publisher={Elsevier},
  doi={10.1016/0378-4371(86)90095-6}
}

@misc{kammerlander2018zeroth,
      title={The zeroth law of thermodynamics is redundant}, 
      author={Philipp Kammerlander and Renato Renner},
      year={2018},
      eprint={1804.09726},
      archivePrefix={arXiv},
      primaryClass={math-ph},
      url={https://arxiv.org/abs/1804.09726}, 
}

@article{LebLiebSimon,
    author = { Lebowitz, Joel L. and Lieb, Elliott H. and {(appendix by B.Simon)}},
    title = {The Constitution of Matter: Existence of Thermodynamics for
Systems Composed of Electrons and Nuclei},
    journal = {Advances in Math.},
    year = {1972},
    vol={9},
page={316},
doi= {10.1016/0001-8708(72)90023-0}
}

@article{CanonicalTypicality,
  title = {Canonical Typicality},
  author = {Goldstein, Sheldon and Lebowitz, Joel L. and Tumulka, Roderich and Zanghi, Nino},
  journal = {Phys. Rev. Lett.},
  volume = {96},
  issue = {5},
  pages = {050403},
  numpages = {3},
  year = {2006},
  month = {Feb},
  publisher = {American Physical Society},
  doi = {10.1103/PhysRevLett.96.050403},
  url = {https://link.aps.org/doi/10.1103/PhysRevLett.96.050403}
}

@Article{Popescu2006,
author={Popescu, Sandu
and Short, Anthony J.
and Winter, Andreas},
title={Entanglement and the foundations of statistical mechanics},
journal={Nature Physics},
year={2006},
month={Nov},
day={01},
volume={2},
number={11},
pages={754-758},
issn={1745-2481},
doi={10.1038/nphys444},
url={https://doi.org/10.1038/nphys444}
}

@Article{Haag1974,
author={Haag, Rudolf
and Kastler, Daniel
and Trych-Pohlmeyer, Ewa B.},
title={Stability and equilibrium states},
journal={Communications in Mathematical Physics},
year={1974},
month={Sep},
day={01},
volume={38},
number={3},
pages={173-193},
issn={1432-0916},
doi={10.1007/BF01651541},
url={https://doi.org/10.1007/BF01651541}
}

@incollection{gorini1985quantum,
  title={Quantum Gibbs states and the Zeroth law of thermodynamics},
  author={Gorini, Vittorio and Frigerio, Alberto and Verri, Maurizio},
  booktitle={Quantum Probability and Applications II},
  pages={240--247},
  year={1985},
  publisher={Springer},
  doi={10.1007/BFb0074476}
}

@book{thirring1983course,
  title={A Course in Mathematical Physics: Volume 4: Quantum Mechanics of Large Systems},
  author={Thirring, Walter},
  year={1983},
  publisher={Springer Science \& Business Media},
  doi={10.1007/978-3-7091-7526-2}
}

@book{OperatorTheory,
  title={Unbounded Self-adjoint Operators on Hilbert Space},
  author={Schmüdgen, Konrad},
  year={2012},
  publisher={Springer Dordrecht},
  doi={https://doi.org/10.1007/978-94-007-4753-1}
}

@book{Gibbs1902,
    author = {Gibbs, Josiah W.},
    title = {Elementary Principles in Statistical Mechanics},
    publisher = {Yale University Press},
    year = {1902}
}

@article{Lenard1978,
  author    = {Lenard, Andrew},
  journal   = {Journal of Statistical Physics},
  title     = {Thermodynamical proof of the Gibbs formula for elementary quantum systems},
  year      = {1978},
  issn      = {1572-9613},
  month     = dec,
  number    = {6},
  pages     = {575--586},
  volume    = {19},
  doi       = {10.1007/bf01011769},
  publisher = {Springer Science and Business Media LLC},
}

@article{PuszWoronwicz,
  author    = {Pusz, Wiesław  and Woronowicz, Stanisław L.},
  journal   = {Commun.Math. Phys.},
  title     = {Passive states and KMS states for general quantum systems},
  year      = {1978},
  issn      = {},
  month     = {},
  number    = {58},
  pages     = {273–290},
  volume    = {},
  doi       = {10.1007/BF01614224},
  publisher = {Springer},
}

@article{Boltzmann1877,
  author  = {Boltzmann, Ludwig},
  title   = {{\"U}ber die Beziehung zwischen dem zweiten Hauptsatze der
             mechanischen W{\"a}rmetheorie und der Wahrscheinlichkeitsrechnung
             respektive den S{\"a}tzen {\"u}ber das W{\"a}rmegleichgewicht},
  journal = {Sitzungsberichte der Kaiserlichen Akademie der Wissenschaften,
             Mathematisch-Naturwissenschaftliche Classe},
  volume  = {76},
  pages   = {373--435},
  year    = {1877},
  note    = {Reprinted in \emph{Wissenschaftliche Abhandlungen}, Vol.~II, 
             pp.~164--223, Barth, Leipzig, 1909; English transl. in 
             \emph{Entropy} 17, 1971--2009 (2015)}
}
